\documentclass[journal,12pt,onecolumn,draftclsnofoot,]{IEEEtran}
\usepackage{amsmath,amsfonts,amssymb}
\usepackage{algorithmic}
\usepackage{algorithm}
\usepackage{amsthm}
\usepackage{array}
\usepackage[caption=false]{subfig}
\usepackage{textcomp}
\usepackage{stfloats}
\usepackage{url}
\usepackage{verbatim}
\usepackage{xcolor} 
\usepackage{empheq}
\usepackage{graphicx}
\usepackage[noadjust]{cite}
\usepackage[thinc]{esdiff}
\newtheorem{theorem}{Theorem}

\newtheorem{definition}{Definition}

\begin{document}

\title{Age of Information in Unreliable Tandem Queues}

\author{\IEEEauthorblockN{
	Muthukrishnan Senthil\,Kumar, {\em Member, IEEE}, 
    Aresh Dadlani, {\em Senior Member, IEEE},  
    and
	Hina Tabassum, {\em Senior Member, IEEE}
    \thanks{M. S.~Kumar is with the Department of Applied Mathematics and Computational Sciences, PSG College of Technology, Coimbatore 641004, India (e-mail: msk.amcs@psgtech.ac.in).}%
    \thanks{A.~Dadlani is with the Department of Mathematics and Computing, Mount Royal University, Calgary, AB, Canada (email: adadlani@mtroyal.ca).}%
    \thanks{H.~Tabassum is with the Department of Electrical Engineering and Computer Science, Lassonde School of Engineering, York University, Toronto, ON, Canada (e-mail: hinat@yorku.ca).}
    }
}

\maketitle

\begin{abstract}
Stringent demands for timely information delivery, driven by the widespread adoption of real-time applications and the Internet of Things, have established the age of information (AoI) as a critical metric for quantifying data freshness. Existing AoI models often assume multi-hop communication networks with fully reliable nodes, which may not accurately capture scenarios involving node transmission failures. This~paper presents an analytical framework for two configurations of tandem queue systems, where status updates generated by~a single sensor are relayed to a destination monitor through~unreliable intermediate nodes. Using the probability generating function, we first derive the sojourn time distribution for an infinite-buffer M/M/1 tandem system with two unreliable nodes. We then extend our analysis to an M/G/1 tandem system~with an arbitrary number of unreliable nodes, employing the supplementary variable technique while assuming that only the first node has an infinite buffer. Numerical results demonstrate the impact of key system parameters on the average AoI in unreliable tandem queues with Markovian and non-Markovian service times. 
\end{abstract}


\raggedbottom
\section{Introduction}
\label{intro}
The rise of real-time applications has reshaped communication network requirements in critical domains like autonomous vehicles, industrial automation, and healthcare monitoring. Traditional metrics such as throughput and latency, while essential, overlook the \textit{freshness} of information at the receiver. To address this, the age of information (AoI) has been introduced as an end-to-end measure to quantify staleness from the perspective of the destination node~\cite{Yates2021}. By definition, a status update packet with timestamp \( u \) has an age of \( t - u \) at any time \( t \geq u \). An update is considered \textit{fresh} when its timestamp equals \( t \), yielding an age of zero. At time \( t \), the age of the most recent update at the destination with timestamp \( u(t) \) is determined by the stochastic process \( \Delta(t) \triangleq t - u(t) \).

Recent research has focused on analyzing and optimizing the average AoI (AAoI) and the peak AoI (PAoI) in queuing models involving sensors sampling data from a physical phenomenon. Early efforts studied AAoI in single-source M/M/1 queues~\cite{Kaul2012}, with subsequent extensions to multiple sources and queuing disciplines~\cite{SenthilKumar2024}. Despite these efforts, the analysis of \textit{multi-hop} wireless network models remains limited but critical for distributed systems and scheduling policies~(\cite{Moradian2021, Tripathi2023}). Tandem queues introduce challenges such as interdependent dynamics, scheduling constraints, and heterogeneous service rates, making AoI analysis crucial for developing~protocols for timely information delivery. 

The study of AoI in multi-hop queuing models has garnered significant attention, with various approaches addressing distinct network configurations. In~\cite{Chiariotti2021}, the authors analyzed the PAoI in a tandem system comprising two interconnected satellite links, while~\cite{Chiariotti2022} established upper and lower bounds on the AAoI under various queuing policies, modeling the communication links as M/M/1 queues. For cache-enabled industrial IoT networks, \cite{ZhengHuiErnest2025} derived the PAoI violation probability, where sensor updates are transmitted to a monitor via a base station (BS) and queried by a cloud server (CS). The authors model the system as a Jackson queueing network, representing the sensors, BS, and CS as interconnected M/M/1 queues with infinite buffer sizes. Aiming to evaluate AAoI,~\cite{Gu2021} examined optimal scheduling policies for resource-constrained multi-hop status update systems by formulating a constrained Markov decision process for a two-hop setting without queuing at intermediate nodes.
Moreover, \cite{Kam2022} explored systems with single- and infinite-capacity tandem queues, employing the stochastic hybrid systems (SHS) framework for the former and a queueing-theoretical approach for the latter. Though effective for small-scale systems, SHS faces scalability challenges as network complexity grows. To address bufferless tandem networks, \cite{sinha2024} adopted a recursive analytical approach to evaluate mean PAoI under~preemptive and non-preemptive policies.

While insightful, previous works~(\cite{Chiariotti2021, Chiariotti2022, ZhengHuiErnest2025, Gu2021, Kam2022, sinha2024}) primarily focus on specific configurations, such as preemptive or infinite capacity queues with reliable transmission, limiting their applicability to broader scenarios where nodes may fail, thus resulting in packet losses~\cite{Kumar2023}. In this paper, we analyze the AAoI in a single-source status update network with \textit{unreliable} tandem queues, deriving closed-form expressions for two system configurations: (i) an unreliable M/M/1 tandem queuing system with infinite-capacity buffers and (ii) an unreliable non-Markovian bufferless tandem queuing system. The former is analyzed using the Laplace–Stieltjes transform (LST) of sojourn time distribution and the probability generating function (PGF) of joint stationary queue length distribution, whereas the latter employs the supplementary variable technique to characterize the sojourn time distribution, with a generalization to arbitrary number of queues in tandem.

\begin{table*}[!t]
\centering
\caption{Balance Equations of Two M/M/1 Nodes in Tandem\vspace{-0.2cm}}
\label{tab1}
\renewcommand{\arraystretch}{1.5}
\begin{tabular}{l p{0.77\columnwidth}}
\hline
\hline
\textbf{State} & \textbf{Balance Equation} \\
\hline
$(0,0,0)$ & $(\lambda+\alpha) q_0(0,0) = \mu_2 q_0(0,1) + \gamma q_1(0,0)$ \\
$(0,0,k), \, k>0$ & $(\lambda+\mu_2+\alpha)q_0(0,k) = \mu_1 q_0(1,k-1) + \mu_2 q_0(0,k+1) + \gamma q_1(0,k)$ \\
$(0,n,k), \, n>0, k \geq 0$ & $(\lambda+\mu_1+\mu_2+\alpha)q_0(n,k) = \lambda q_0(n-1,k) + (1 - \delta_{k,0})\mu_1 q_0(n+1,k-1) + \mu_2 q_0(n,k+1) + \gamma q_1(n,k)$ \\
$(1,n,k), \, n,k>0$ & $(\lambda+\gamma)q_1(n,k) = \alpha q_0(n,k) + \lambda q_1(n-1,k)$ \\
\hline
\end{tabular}

\vspace{0.3em}
\!\!\!\!\!\!\!\!\!\!\!\!\!\!\!\!\!\!\!\!\!\!\!\!\!\!\!\!\!\!\!\!\!\!\!\!\!\!\!\!\!\!\!\!\!\!\!\!\!\!\!\!\!\!\!\!\!\!\!\!\!\!\!\!\!\!\!\!\!\!\!\!\!\!\!\!\!\!\!\!\!\!\small\textit{Note:} \(\delta_{k,0}\) is the Kronecker delta function, equal to 1 when \(k = 0\) and 0 otherwise.
\end{table*}
\section{System Model and Assumptions}
\label{sec2}
We model a multi-hop wireless network as a series of tandem queues, each with a unit capacity, as illustrated in~\figurename{~\ref{figure1}}. Real-time updates are generated by a single-antenna sensor following a Poisson process with rate \(\lambda\) and are forwarded to the queue at the first node. The update packets then traverse the network sequentially, with each queue \(i\) processing them at a service rate \(\mu_i > 0\) before reaching the monitor. The model accounts for potential transmission failures due to server breakdowns. Following the approach in~\cite{Yates2021} and \cite{Kumar2023}, we aim to derive the AAoI, defined as $\Delta = \lim_{T \rightarrow \infty} \frac{1}{T} \int_{0}^{T} \Delta(t) \,dt$.\vspace{0.05cm}

We begin our analysis by examining two unreliable M/M/1 queues with infinite capacity, where the entire tandem network experiences breakdowns at a rate of $\alpha > 0$. Upon failure, the network undergoes repair, with the average repair time given by $1/\gamma$.
\begin{figure}[!t]
	\centering
	\includegraphics[width=0.8\linewidth]{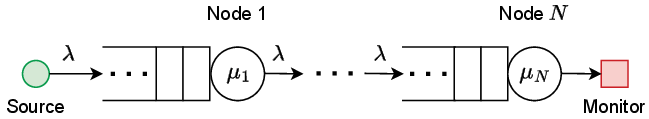}
	\vspace{-0.6cm}
	\caption{The M/M/1 tandem queuing model for \( N \) unreliable nodes.}
	\label{figure1}
\end{figure}

\section{AAoI of Unreliable M/M/1 Tandem Queues}
\label{sec3}
To model a tandem system with \( N = 2 \) unreliable Markovian nodes, we define an irreducible continuous-time Markov chain (CTMC), \( \mathcal{X}(t) = \{(X_0(t), X_1(t), X_2(t)) \mid t \geq 0\} \), where \( X_0(t) \) represents the global operational state of the tandem network, and \( X_1(t) \) and \( X_2(t) \) denote the number of packets at node 1 and node 2, respectively. The network is operational when \( X_0(t) \!=\! 0\) and under repair when \( X_0(t) \!=\! 1 \). The state transition probabilities are given as follows, with the corresponding states and balance equations given in Table~\ref{tab1}:
\begin{equation}
p_{i,n_1,n_2}(t) = \Pr[X_0(t) = i, X_1(t) = n_1, X_2(t) = n_2],  \\
\end{equation}
such that $\lim_{t \to \infty} p_{i,n_1,n_2}(t) = q_i(n_1,n_2)$ for $i \in\{0,1\}$ and $n_1,n_2 \geq 0$. To solve the balance equations in Table~\ref{tab1}, along with the normalization condition \(\sum_{n_1,n_2} [q_0(n_1,n_2) + q_1(n_1,n_2)] = 1\), we employ the PGF approach defined below.
\begin{definition}
\label{def1}
Let $z_1$ and $z_2$ be auxiliary complex variables associated with the number of packets at node~1 and node~2, respectively. The PGF encodes the joint stationary queue length distribution defined as follows, where $|z_1| < 1$ and $|z_2| < 1$:
\begin{eqnarray}
\Pi_i(z_1, z_2) = \sum_{n=0}^{\infty} \sum_{k=0}^{\infty} q_i(n, k) z_1^n z_2^k, \quad i \in \{0,1\}. 
\end{eqnarray}
\end{definition}

From Definition~\ref{def1} and the balance equations in Table~\ref{tab1}, we derive the following system of functional equations, where \( D(z_1, z_2) = z_1 z_2[\lambda(1-z_1) + \gamma] + \mu_1 z_2(z_1 - z_2) + \mu_2 z_1(z_2 - 1) \), \( A(z_1, z_2) = \mu_2 z_1(z_2 - 1) + \mu_1 z_2(z_2 - z_1) \), \( B(z_1, z_2) = -A(z_1, z_2) \), and \( C(z_1, z_2) = \mu_1 z_2(z_2 - z_1) + \mu_2 z_1(z_2 - 1) \):
\begin{subequations}
  \begin{empheq}[left=\empheqlbrace]{align}
    &D(z_1,z_2)\Pi_0(z_1,z_2) = A(z_1,z_2) \Pi_0 (z_1,0) \nonumber \\
    &\qquad\qquad\qquad\qquad\quad\,\, + B(z_1,z_2) \Pi_0(0,z_2) \nonumber \\
    &\qquad\qquad\qquad\qquad\quad\,\, + C(z_1,z_2) \Pi_0(0,0) \nonumber \\
    &\qquad\qquad\qquad\qquad\quad\,\, + \gamma z_1 z_2 \Pi_1(z_1,z_2), \label{eq3a}\\
    &\Pi_1(z_1,z_2) = \frac{\alpha}{\lambda(1-z_1) + \gamma}\Pi_0(z_1,z_2), \label{eq3b}
  \end{empheq}
  \label{eq3}
  \vspace{-0.1cm}
\end{subequations}
By substituting \eqref{eq3b} into \eqref{eq3a} and performing some algebraic manipulations, we obtain:
\begin{align}
    &\Pi_0(z_1,z_2) \left[\frac{\alpha D(z_1,z_2)}{\lambda(1-z_1)+\gamma}-\alpha \gamma z_1 z_2\right] 
    = \nonumber \\
    &\qquad \frac{\alpha}{\lambda(1-z_1)+\gamma} \left[ A(z_1,z_2) \Pi_0 (z_1,0) + B(z_1,z_2)\Pi_0(0,z_2) \right. \nonumber \\
    &\qquad \left. +\, C(z_1,z_2) \Pi_0(0,0) \right].
    \label{eq4}
\end{align}
\begin{theorem}
    The LST of the sojourn time distribution of the proposed system is given by:
    \begin{align}
        &W^*(s) = P\left(1 - \frac{s}{\lambda}\right)  \nonumber \\
           &= \frac{\frac{\alpha}{\lambda(1 - f(1 - \frac{s}{\lambda})) + \gamma} 
           \Big[ C\big(f(1 \!-\! \frac{s}{\lambda}), (1 \!-\! \frac{s}{\lambda})\big) \Pi_0(0,0) \Big]}
           {\frac{\alpha}{\lambda(1 - f(1 - \frac{s}{\lambda})) + \gamma} 
           D\big(f(1 \!-\! \frac{s}{\lambda}), (1 \!-\! \frac{s}{\lambda})\big) 
           \!- \alpha \gamma f(1 \!-\! \frac{s}{\lambda}) (1 \!-\! \frac{s}{\lambda})},\vspace{-0.2em}
    \end{align}
where $f\left(1 - \frac{s}{\lambda}\right) = \mu_1 (1 - \frac{s}{\lambda})^2/\left(\mu_1 + \mu_2 \left(\frac{s}{\lambda}\right)\right)$.
\end{theorem}
\begin{proof}
From \eqref{eq3b}, it is clear that $\Pi_0(1,1) + \Pi_1(1,1) = 1$. This implies that $\Pi_0(1,1) = \gamma/(\alpha+\gamma)$ and $\Pi_1(1,1) = \alpha/(\alpha+\gamma)$. Now, substituting $z_1 = f(z_2) = \mu_1 z_2^2 /\left(\mu_1 + \mu_2 (1 - z_2)\right)$ in \eqref{eq4} yields the marginal PGF of the stationary queue length distribution of node 2 as follows:
\begin{align}
    P(z_2) &= \Pi_0(f(z_2),z_2) \nonumber \\
    &= \frac{\frac{\alpha}{\lambda(1 - f(z_2)) + \gamma} 
    \big[C(f(z_2),z_2) \Pi_0(0,0)\big]}
    {\frac{\alpha}{\lambda(1 - f(z_2)) + \gamma} 
    D(f(z_2),z_2) - \alpha \gamma f(z_2)z_2}\, ,
    \label{eq6}
\end{align}
since \( A(z_1,z_2) \!= B(z_1,z_2) \!= 0 \). From \eqref{eq6}, as \( z_2 \to 1 \), we get:
\begin{equation}
    \Pi_0(0,0) = \frac{\gamma}{\alpha+\gamma} - \lambda \left[\frac{1}{\mu_1}+\frac{1}{\mu_2} \right].
    \label{eq7}
\end{equation}
Equation~\eqref{eq7} represents the stability condition of the unreliable two-node tandem queueing system as:
\begin{equation}
    \lambda \left[\frac{1}{\mu_1}+\frac{1}{\mu_2} \right] < \frac{\gamma}{\alpha+\gamma}.
\end{equation}
Finally, from \eqref{eq6}, the LST of the sojourn time distribution~of the system, as derived using~\cite{peter2006}, is \(  W^*(s) = P\left(1 - \frac{s}{\lambda} \right) \).
\end{proof}
\vspace{0.3em}

Using the invariant relation~\cite{ino2017} among the AoI, PAoI, and system delay distributions, the LST of the AoI distribution is derived to be:
\begin{equation}
    \Delta^*(s) = \frac{\lambda \big[ W^*(s) \!- W^*(s) h^*(s) \!+W^*(s \!+\! \lambda) \frac{s h^*(s)}{s+\lambda} \big]}{s},
    \label{eq9}
\end{equation}
where \( h^*(s) = \left(\frac{\alpha+\gamma}{\gamma}\right)\! \left(\frac{\mu_1}{s+\mu_1}\right) \!\left(\frac{\mu_2}{s+\mu_2}\right) \). Hence, the AAoI is:
\begin{equation}
    \Delta = -\frac{d \Delta^*(s)}{ds} \Big|_{s \to 0}\, .
    \label{eq10}
\end{equation}

Due to the complexity of deriving the sojourn time distribution in an \(N\)-node unreliable tandem queue, we numerically analyze its impact in our simulations. 

\begin{figure}[!t]
	\centering
	\includegraphics[width=0.8\linewidth]{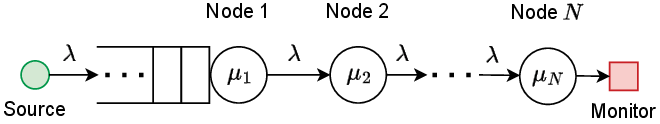}
	\vspace{-0.6cm}
	\caption{The M/G/1 tandem queuing model for \( N \) unreliable nodes, where all \( (N-1) \) nodes but the first are bufferless.}
	\label{figure2}
\end{figure}
\section{AoI of Unreliable M/G/1 Tandem Queues}
\label{sec4}
The unreliable M/G/1 tandem queuing model with bufferless intermediate nodes shown in \figurename{~\ref{figure2}} is analyzed in this section. For \(N=2\), node~1 has an infinite buffer, while the subsequent tandem node is bufferless. Here, the processing time of each packet is an i.i.d. random variable with generally distributed tandem service times \( H_i \), for \( i \in \{1,2\} \),  and corresponding distribution functions \( F_{H_i}(t) \). The LST of \( H_i \) is denoted by \( s_i^*(s) \). When a packet completes service at node one, it moves to node two only if node two is idle. Otherwise, it remains at node one and blocks it until node two becomes available. During this time, no new packet can enter service at node one. Each node \(i\) in the two-node tandem queuing system is subject to breakdowns while serving packets, with a constant failure rate \( \alpha_i = \alpha \). Upon failure, node \(i\) undergoes a repair process with a generally distributed repair time \( R_i \), having distribution function \( F_{R_i}(t) \) and LST \( r_i^*(s) \). We denote the remaining service and repair times at time \( t \) by \( S_i^0(t) \) and \( R_i^0(t) \), respectively. The server state \( X(t) \) is defined as:
\begin{align}
    X(t) &= 
    \begin{cases}
        0, & \text{if the server is idle},  \\
        1, & \text{if node 1 is serving a packet}, \\
        2, & \text{if node 1 has failed and is under repair}, \\
        3, & \text{if node 2 is serving a packet}, \\
        4, & \text{if node 2 has failed and is under repair}.
    \end{cases}
    \label{eq11}
\end{align}

We now define the CTMC representing the system state as \(\mathcal{K}(t) \!=\! \{ (X(t), L(t), S_1^0(t), S_2^0(t), R_1^0(t), R_2^0(t)) \mid t \geq 0 \}\), where \( L(t) \) denotes the number of packets queued at node~1. Based on \eqref{eq11}, the state probabilities are defined as follows, where \( p_0(t) \) is the probability that the server is idle at time \( t \), \( p_{k,n}(x,t)dx \) is the joint probability that node~1 (if \( k=1 \)) or node~2 (if \( k=3 \)) is busy transmitting a packet during the remaining service time \( (x, x+dx) \), and \( p_{j,n}(x,y,t)dy \) is the joint probability that failed node~1 (if \( j=2 \)) or failed node~2 (if \( j=4 \)) is undergoing repair within the remaining repair time \( (y, y+dy) \):
\begin{equation}
\begin{aligned}
    &\left\{
    \begin{aligned}
        &p_0(t) = \Pr[X(t)=0, L(t)=0], \\
        &p_{1,n}(x,t)dx = \Pr[X(t) \!=\! 1, L(t) \!=\! n, x <\! S_1^0(t) \!< x \!+\! dx], \\
        &p_{2,n}(x,y,t)dy = \Pr[X(t)=2, L(t)=n, S_1^0(t)=x,\\
        &\qquad\qquad\qquad\qquad\quad\! y < R_1^0(t) < y+dy], \\
        &p_{3,n}(x,t)dx = \Pr[X(t) \!=\! 3, L(t) \!=\! n, x <\! S_2^0(t) \!< x \!+\! dx], \\
        &p_{4,n}(x,y,t)dy = \Pr[X(t)=4, L(t)=n, S_2^0(t)=x,\\
        &\qquad\qquad\qquad\qquad\quad\! y < R_2^0(t) < y+dy].
    \end{aligned} 
    \right.
\end{aligned}
\label{eq12}
\end{equation}
 Using the supplementary variable technique, we obtain the following balance equations as \( t \to \infty \), where \(\forall i \in \{1,2\}, f_{s_i}(\cdot) \) and \( f_{r_i}(\cdot) \) denote the probability density functions of the \( i\)-th service and repair times, respectively:
\begin{equation}
\begin{aligned}
    &\left\{
    \begin{aligned}
        &\lambda p_0 = p_{1,0}(0) + p_{3,0}(0), \\
        &\frac{dp_{1,0}(x)}{dx} = (\lambda + \alpha_1)p_{1,0}(x) - (\lambda p_0) f_{s_1}(x) - p_{2,0}(x,0), \\
        &\frac{dp_{1,n}(x)}{dx} = (\lambda + \alpha_1)p_{1,n}(x) - \lambda p_{1,n-1}(x) - p_{2,n}(x,0), \\
        &\frac{\partial p_{2,n}(x,y)}{\partial y} \!=\! \lambda p_{2,n}(x,y) \!-\! \lambda p_{2,n\!-1}(x,y) \!-\! \alpha_1 p_{1,n}(x) f_{r_1\!}(y), \\
        &\frac{dp_{3,n}(x)}{dx} = (\lambda + \alpha_2)p_{3,n}(x) - \lambda p_{3,n-1}(x)\\
        &\qquad\qquad\quad - p_{1,n}(0) f_{s_2}(x) - p_{4,n}(x,0), \\
        &\frac{\partial p_{4,n}(x,y)}{\partial y} \!=\! \lambda p_{4,n}(x,y) \!-\! \lambda p_{4,n\!-1}(x,y) \!-\! \alpha_2 p_{3,n}(x) f_{r_2\!}(y),
    \end{aligned}
    \right.
\end{aligned}
\label{eq13}
\end{equation}
The normalizing condition is given by:
\begin{align}
    p_0 &+ \sum_{n=0}^{\infty} \left[ \int_0^{\infty}\!\! \big( p_{1,n}(x) + p_{3,n}(x) \big) dx \right]\qquad \nonumber \\
    &+ \sum_{n=0}^{\infty} \left[ \int_0^{\infty}\!\! \int_0^{\infty}\!\! \big( p_{2,n}(x,y) + p_{4,n}(x,y) \big) dx \, dy \right] = 1. 
    \label{eq14}
\end{align}

\begin{figure*}[t]
    \centering
  \subfloat[AAoI vs. $\lambda$ ($N=2$).\label{fig3a}]{%
       \includegraphics[width=0.3365\linewidth]{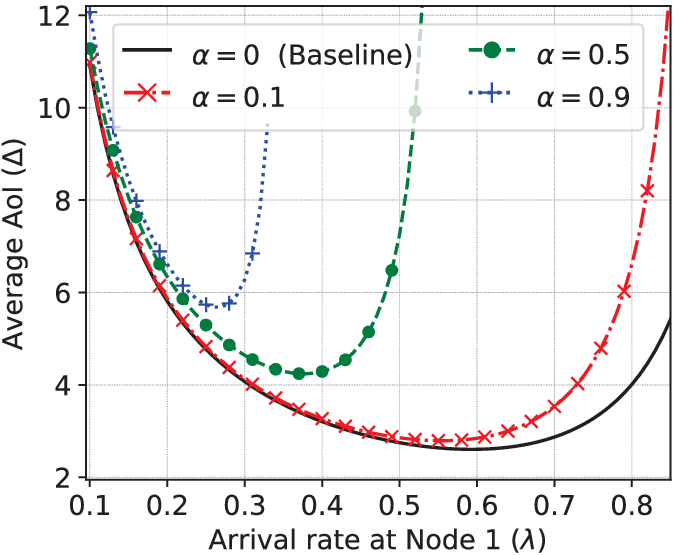}}
    \hfill
  \subfloat[AAoI vs. $\lambda$ ($N=4$).\label{fig3b}]{%
        \includegraphics[width=0.3195\linewidth]{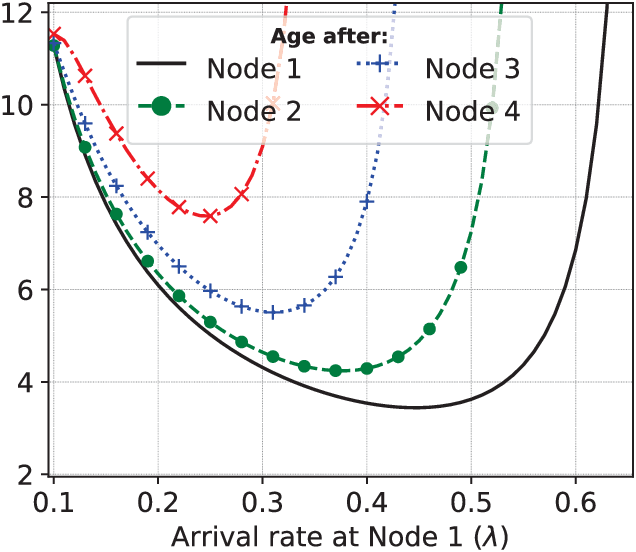}}
        \hfill
    \subfloat[Per-node delay vs. $\lambda$ ($N=4$).\label{fig3c}]{%
       \includegraphics[width=0.343\linewidth]{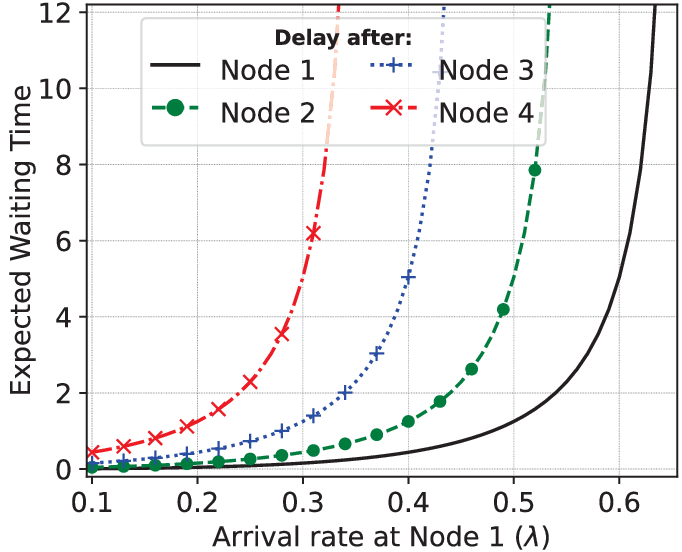}}
    \vspace{0.2em}
  \caption{Average age in an unreliable M/M/1 tandem queue with infinite capacity: (a) $N=2$ with $\mu_1=\mu_2=1$, $\gamma=1$; (b) $N=4$ with $\mu_i=1$, $\forall i \in \{1,2,3,4\}$, $\gamma=1$, $\alpha=0.5$; (c) Expected waiting time for $N=4$ with $\mu_i=1$, $\forall i \in \{1,2,3,4\}$, $\gamma=1$, $\alpha=0.5$.}
  \label{fig3n} 
\end{figure*}
Denoting the LSTs of $p_{k,n}(x)$ and $p_{j,n}(x,y)$ by $p_{k,n}^*(\theta)$ and $p_{j,n}^*(\theta, s)$ for $k \in \{1,3\}$ and $j \in \{2,4\}$ respectively, the marginal PGFs can be expressed as:
\begin{equation}
\begin{aligned}
    &\left\{
    \begin{aligned}
        &p_{k}^*(z,\theta) = \sum_{n=0}^{\infty} p_{k,n}^*(\theta) z^n, \\
        &p_{k}(z,0) = \sum_{n=0}^{\infty} p_{k,n}(0) z^n, \quad k=1,3, \\
        &p_{j}^*(z,\theta,s) = \sum_{n=0}^{\infty} p_{j,n}^*(\theta,s) z^n, \\
        &p_{j}(z,\theta,0) = \sum_{n=0}^{\infty} p_{j,n}(\theta,0) z^n, \quad j=2,4.
    \end{aligned}
    \right.
\end{aligned}
\end{equation}

Applying LST to both sides of \eqref{eq13} and after some algebraic manipulations, the PGF of the system size is obtained as:
\begin{equation}
    P(z) = \frac{h^*(\phi(0,z)) (1 - z) p_0}{h^*(\phi(0,z)) - z},
    \label{eq16}
\end{equation}
where \(\forall i \in \{1,2\}, \alpha_i = \alpha\), \(\phi(s,z) = s + \lambda - \lambda z +~\alpha - \alpha r^*(s+\lambda-\lambda z) \), \( h^*(s) = s_1^*(s) s_2^*(s) \), and \( p_0 = 1 - \lambda \left( 1 + \frac{\alpha}{\gamma} \right) \left( \frac{1}{\mu_1} + \frac{1}{\mu_2} \right) \) when \( r_1^*(s) = r_2^*(s) = r^*(s) \). Using~\eqref{eq16} and the approach in~\cite{peter2006}, the LST of the sojourn time distribution for an unreliable two-node tandem queueing system is given by \( W^*(s) \!=\! P\left(1-\frac{s}{\lambda}\right)\). In general, we have:
\begin{equation}
        W^*(s) = P\left(1 - \frac{s}{\lambda}\right) 
        = \frac{h^*\left(\phi\left(0, 1 - \frac{s}{\lambda}\right)\right)\left(\frac{s}{\lambda}\right)p_0}{h^*\left(\phi\left(0, 1 - \frac{s}{\lambda}\right)\right) - 1 + \frac{s}{\lambda}}\, ,
        \label{eq17}
    \end{equation}
where $h^*(s) = \prod_{i=1}^{k} f_{s_i}(s)$ and $f_{s_i}(s)$ is the \(i\)-th service time distribution. Using~\eqref{eq16}, the LST of the age distribution for \(N\) unreliable tandem queues is derived in the following theorem.
\begin{theorem}
    The LST of the age distribution at any node $i$ is:
    \begin{align}
        \Delta^*(s) = W^*(s) - \frac{s p_0 h^*(s)}{s + \lambda h^*(s+\lambda)}.
    \label{eq18}
\end{align}
\end{theorem}
\begin{proof}
    Let \( Y_k \) denote the inter-arrival time between the \(k\)-th and \((k-1)\)-th packets, and \( T_{i,k} \) the transmission time of the \(k\)-th packet at node \(i\). The \(n\)-th packet, generated at time \(t_n\), is served at time \( t_n' \) at queue~1, with its age ending at \( t_n'' \). Here, we assume \(t_0 = 0\). The age reduction equals the difference between the ages of consecutive packets, which represents the inter-arrival time at queue 1. For $N=2$, the accumulated age is \( Y_1 \times T_1 \), where \( T_1 = t_1'' - t_1 \) is the system time at node 2. Node failures~occur at \( r_i \) during service at node \(i\), followed by repairs over \((r_i, r_i')\). Let \( \{A_t\}_{t \geq 0} \) be the AoI process. A sample path of this process is defined by \( \{(t_n, X_n)\}_{n \geq 0} \), where \( X_n = A_{t_n} \) is the AoI immediately after the \(n\)-th update. For $n \geq 1$ and \(t \in (t_{n-1}, t_n)\), \( A_t \triangleq X_{n-1} + (t - t_{n-1}) \). 
    Now, consider \( M_t = \sup\{n \in \{0,1,2,\ldots\} : t_n \leq t\} \). Using \( \lim_{n \to \infty} t_n/n = 1/\lambda, \lambda \in (0, \infty) \), the AoI density for \( x \geq 0 \) is calculated as follows:
    \begin{align}
        &A(x) = \lim_{T \to \infty} \frac{1}{T} \int_0^T \mathbf{1}_{A_t \leq x} \, dt 
             = \lim_{T\rightarrow \infty}\frac{1}{T}\sum_{n=0}^{M_T-1}T_n(x) \nonumber \\
        &=\lim_{T\rightarrow \infty}\frac{1}{T}\! \sum_{n=0}^{M_T-1}\!\! \int_0^x\! \textbf{1}_{X_n\leq u}du -\! \int_0^x\!\! \textbf{1}_{\max(X_{n-1},Y_n)+S_n\leq u} du \nonumber \\
        &= \lambda \left( W(x)-P_1(x) \right).
        \label{eq20}
    \end{align}
    where \(S_n\) denotes the service time of the \(n\)-th packet, \( W(x) \) is the density of the sojourn time distribution, and \( P_1(x) \) is~the density of the PAoI
    Applying LST to~\eqref{eq20} yields \eqref{eq18}, which completes the proof.
\end{proof}
Substituting~\eqref{eq18} into~\eqref{eq10} results in the AAoI at node \(i\) for the general case of unreliable M/G/1 queues in tandem.


\section{Numerical Results and Discussions}
\label{sec5}
In this section, we assess the analytical findings on AAoI for the proposed unreliable tandem queue models. All evaluations are benchmarked against failure-free models ($\alpha = 0$). Unless stated otherwise, the service rate of node \(i\) is set to \(\mu_i = 1\), with a repair rate of \(\gamma =1\).
    


\figurename{~\ref{fig3n}} compares the AAoI for unreliable M/M/1 tandem queues with infinite capacity as a function of the source update arrival rate. The impact of \( \lambda\) and the network unreliability factor \( \alpha\) on the average age for \(N=2\) is shown in~\figurename{~\ref{fig3a}}. The baseline shows a U-shaped trend, where AAoI initially decreases with increasing \( \lambda\) due to more frequent updates but later increases as queueing delays dominate. As \( \alpha\) increases, representing~a higher failure probability, the minimum achievable AAoI \( (\Delta_{min})\) worsens, and the optimal arrival rate \( (\lambda^*) \) shifts leftward, indicating that the system can sustain lower arrival rates before experiencing instability. Moreover, higher \( \alpha\) also accelerates the onset of performance degradation, leading to steeper AAoI increases at moderate to high \( \lambda\) values. \figurename{~\ref{fig3b}} shows the cumulative effect of queuing and service delays~in the unreliable four-node tandem network. We observe that~the AAoI deteriorates progressively across successive nodes, with later nodes exhibiting a higher \( \Delta_{min}\) and an earlier onset~of instability as \( \lambda\) increases. Notably,~\( \Delta_{min}\) increases by approximately \( 118\%\) as the number of nodes increases from 1 to 4. This trend arises because, while moderate arrival~rates initially reduce AAoI, higher traffic levels intensifies congestion by increasing the expected waiting time at each node as shown in \figurename{~\ref{fig3c}}, resulting in significant performance degradation at downstream unreliable nodes with infinite buffer capacities.
\begin{figure*}[t]
    \centering
  \subfloat[AAoI vs. $\lambda$ ($N=2$, Exp. service time).\label{fig4a}]{%
       \includegraphics[width=0.34\linewidth]{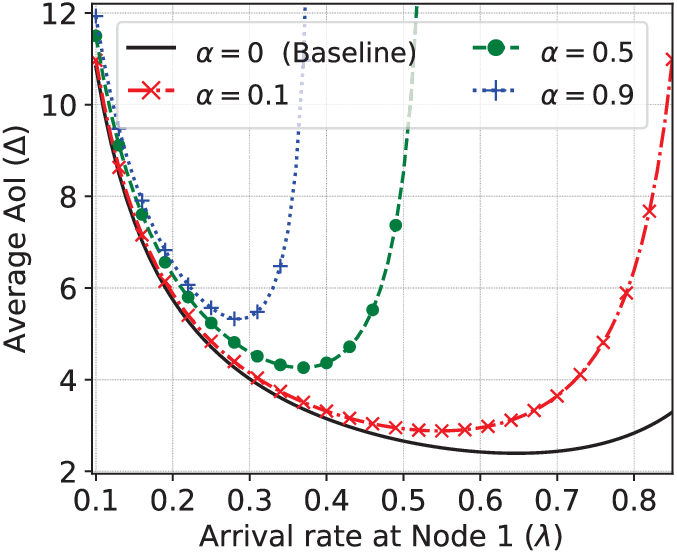}}
    \hfill
  \subfloat[AAoI vs. $\lambda$ ($N=2$, Erl. service time).\label{fig4b}]{%
        \includegraphics[width=0.319\linewidth]{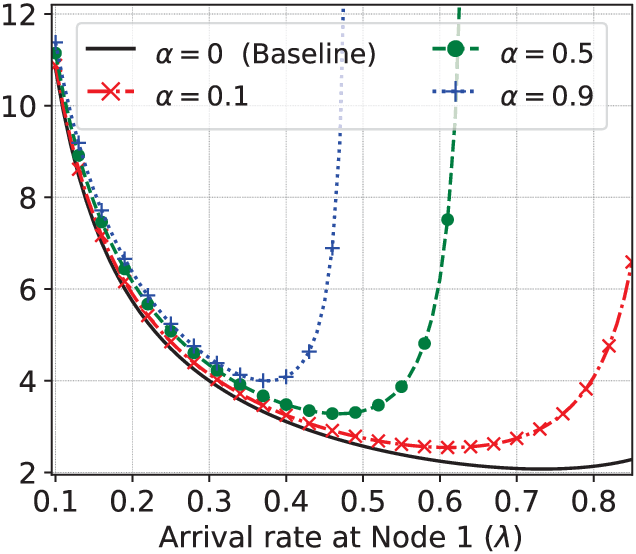}}
        \hfill
    \subfloat[AAoI vs. $\lambda$ ($N=4$, Erl. vs. H$_2$ service times).\label{fig4c}]{%
       \includegraphics[width=0.319\linewidth]{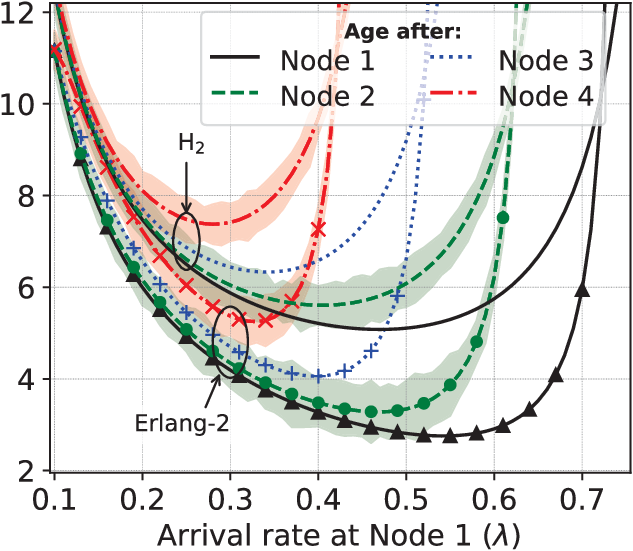}}
    \vspace{0.2em}
  \caption{Average age in an unreliable bufferless M/G/1 tandem queue: (a) $N=2$ with exponential service time distribution where $\mu_1=\mu_2=1$, $\gamma=1$; (b) $N=2$ with Erlang-2 service time distribution where $\mu_1=\mu_2=1$, $\gamma=1$; (c) $N=4$ with Erlang-2 and hyper-exponential of order 2 (H$_2$) service time distributions where $\mu_i=1$, $\alpha_i (=\alpha) = 0.5$, and $\gamma=1$.}
  \label{fig4n} 
\end{figure*}

\figurename{~\ref{fig4n}} presents the AAoI as a function of \(\lambda\) for an unreliable bufferless M/G/1 tandem queue, evaluated under varying but uniform nodal failure rates \(\alpha_i = \alpha\). \figurename{~\ref{fig4a}} shows results~for \(N=2\) with an exponential service time distribution (Exp.), representing the classical M/M/1 queue. In this bufferless setting, update packets move sequentially, with transmissions from node~1 proceeding directly to node~2 without delay. Conversely, in \figurename{~\ref{fig3a}}, where an M/M/1 queue includes a buffer, packets at node~1 must wait for service at node~2,~increasing AAoI due to queuing delays. \figurename{~\ref{fig4b}} examines the impact of adopting the Erlang-2 (Erl.) distribution, which~features a lower coefficient of variation in service times. The key distinction is the reduced AAoI variability, particularly at moderate arrival rates. This is evident in the \(92.8\%\) increase in the minimum average~age as \(\alpha\) rises from $0$ to $0.9$, which remains significantly lower than the \(122.3\%\) increase observed under the Exp. distribution. Finally, \figurename{~\ref{fig4c}} compares the AAoI performance for \(N=4\) between Erl. and hyper-exponential of order 2 (H$_2$) service time distributions, along with the $95\%$ confidence intervals. Unlike H$_2$, Erl. reduces fluctuations, resulting in lower AAoI~at~high arrival rates due to reduced queuing delays. This highlights the advantage of Erl. service in multi-hop systems by improving information freshness through lower variability and a more controlled increase in AAoI at higher arrival rates.

\section{Conclusion}
\label{sec6}
\fontdimen2\font=0.50ex
This paper analyzed the AAoI in unreliable tandem queueing models, considering M/M/1 queues with infinite capacity and M/G/1 queues with bufferless intermediate nodes. Using the PGF approach, closed-form expressions for the sojourn time distribution and AAoI in the M/M/1 case were derived, while the supplementary variable technique extended the analysis to M/G/1 systems with arbitrary unreliable nodes. Numerical results show that service time distribution and node failures significantly impact AAoI, with Erlang-2 service reducing fluctuations compared to the exponential case. Moreover, both tandem queueing models become less stable as \(N\) increases, constraining the minimum achievable AAoI to lower arrival rate ranges. Future work could explore correlated Poisson arrivals from multiple sources and the impact of phase-type service and repair times in compromised tandem queueing systems.

\bibliographystyle{IEEEtran}
\bibliography{IEEEabrv, myref}

\begin{thebibliography}{10}
\providecommand{\url}[1]{#1}
\csname url@samestyle\endcsname
\providecommand{\newblock}{\relax}
\providecommand{\bibinfo}[2]{#2}
\providecommand{\BIBentrySTDinterwordspacing}{\spaceskip=0pt\relax}
\providecommand{\BIBentryALTinterwordstretchfactor}{4}
\providecommand{\BIBentryALTinterwordspacing}{\spaceskip=\fontdimen2\font plus
\BIBentryALTinterwordstretchfactor\fontdimen3\font minus
  \fontdimen4\font\relax}
\providecommand{\BIBforeignlanguage}[2]{{%
\expandafter\ifx\csname l@#1\endcsname\relax
\typeout{** WARNING: IEEEtran.bst: No hyphenation pattern has been}%
\typeout{** loaded for the language `#1'. Using the pattern for}%
\typeout{** the default language instead.}%
\else
\language=\csname l@#1\endcsname
\fi
#2}}
\providecommand{\BIBdecl}{\relax}
\BIBdecl

\bibitem{Yates2021}
R.~D. Yates, Y.~Sun, D.~R. Brown, S.~K. Kaul, E.~Modiano, and S.~Ulukus, ``Age
  of information: An introduction and survey,'' \emph{IEEE J. Sel. Areas
  Commun.}, vol.~39, no.~5, pp. 1183--1210, May 2021.

\bibitem{Kaul2012}
S.~Kaul, R.~Yates, and M.~Gruteser, ``Real-time status: How often should one
  update?'' in \emph{Proc. IEEE Int. Conf. Comput. Commun. (INFOCOM)}, Mar.
  2012, pp. 2731--2735.

\bibitem{SenthilKumar2024}
M.~S. Kumar, A.~Dadlani, O.~Ardakanian, I.~Nikolaidis, and J.~J. Harms, ``Age
  analysis of correlated information in multi-source updating systems with
  {MAP} arrivals,'' \emph{IEEE Commun. Lett.}, vol.~28, no.~7, pp. 1539--1543,
  Jul. 2024.

\bibitem{Moradian2021}
M.~Moradian and A.~Dadlani, ``Average age of information in two-way relay
  networks with service preemptions,'' in \emph{Proc. IEEE Global Commun. Conf.
  (GLOBECOM)}, Dec. 2021, pp. 1--6.

\bibitem{Tripathi2023}
V.~Tripathi, R.~Talak, and E.~Modiano, ``Information freshness in multihop
  wireless networks,'' \emph{IEEE/ACM Trans. Networking}, vol.~31, no.~2, pp.
  784--799, Apr. 2023.

\bibitem{Chiariotti2021}
F.~Chiariotti, O.~Vikhrova, B.~Soret, and P.~Popovski, ``Peak age of
  information distribution for edge computing with wireless links,'' \emph{IEEE
  Trans. Commun.}, vol.~69, no.~5, pp. 3176--3191, May 2021.

\bibitem{Chiariotti2022}
------, ``Age of information in multihop connections with tributary traffic and
  no preemption,'' \emph{IEEE Trans. Commun.}, vol.~70, no.~10, pp. 6718--6733,
  Oct. 2022.

\bibitem{ZhengHuiErnest2025}
T.~Z.~H. Ernest and A.~S. Madhukumar, ``Peak age of information analysis of
  status update strategies in cache-enabled {IIoT} networks,'' \emph{IEEE
  Internet Things J.}, vol.~12, no.~3, pp. 3028--3042, Feb. 2025.

\bibitem{Gu2021}
Y.~Gu, Q.~Wang, H.~Chen, Y.~Li, and B.~Vucetic, ``Optimizing information
  freshness in two-hop status update systems under a resource~constraint,''
  \emph{IEEE J. Sel. Areas Commun.}, vol.~39, no.~5, pp. 1380--1392, May 2021.

\bibitem{Kam2022}
C.~Kam and S.~Kompella, ``On the age of information for non-preemptive queues
  in tandem,'' \emph{Frontiers in Commun. and Netw.}, vol.~3, Nov. 2022.

\bibitem{sinha2024}
A.~Sinha, S.~Singhvi, P.~D. Mankar, and H.~S. Dhillon, ``Peak age of
  information under tandem of queues,'' in \emph{Proc. IEEE Int. Symp. Inf.
  Theory (ISIT)}, Jul. 2024, pp. 951--956.

\bibitem{Kumar2023}
M.~S. Kumar, A.~Dadlani, M.~Moradian, A.~Khonsari, and T.~A. Tsiftsis, ``On the
  age of status updates in unreliable multi-source {M/G/1} queueing systems,''
  \emph{IEEE Commun. Lett.}, vol.~27, no.~2, pp. 751--755, Feb. 2023.

\bibitem{peter2006}
P.~G. Harrison and W.~J. Knottenbelt, \emph{Quantiles of Sojourn Times}, 2006,
  pp. 155--193.

\bibitem{ino2017}
Y.~Inoue, H.~Masuyama, T.~Takine, and T.~Tanaka, ``The stationary distribution
  of the age of information in {FCFS} single-server queues,'' in \emph{Proc.
  IEEE Int. Symp. Inf. Theory (ISIT)}, Jun. 2017, pp. 571--575.

\end{thebibliography}

\end{document}